\documentclass[smallextended]{svjour2}
\usepackage{graphicx}
\usepackage[left=2cm,right=2cm]{geometry}
\usepackage{amssymb}
\usepackage{diagbox}
\usepackage{comment}
\newcommand{\dd}{ {\rm{d}}}
\newcommand{\erf}{ {\rm{erf}}}
\newcommand{\erfc}{ {\rm{erfc}}}
\usepackage[normalem]{ulem}  
\usepackage{color} 
\usepackage{xcolor}
\usepackage{indentfirst}
\usepackage{fix-cm}

\newcommand{\abs}[1]{\left| #1 \right|}
\newcommand{\nhi}{n \!\!\!\!{-}}
\newcommand{\hn}{\nhi}
\newcommand{\hnp}{\nhi'}
\newcommand{\half}[1]{\frac{#1}{2}}
\newcommand{\ii}{\mathrm{i}}
\newcommand{\ee}{\mathrm{e}}
\begin{document}
\title{A modified Euler-Maclaurin formula in 1D and 2D with applications in statistical physics}
\author{Jihong Guo \and Yunpeng Liu$^*$}
\institute{Jihong Guo \at Department of Applied Physics, Tianjin University, Tianjin 300350, P.R.China 
\and Yunpeng Liu \at Department of Applied Physics, Tianjin University, Tianjin 300350, P.R.China
            \and
	    \email{yunpeng.liu@tju.edu.cn}
}
	    \date{Received:date /Accepted:date}
\maketitle
\begin{abstract}
The Euler-Maclaurin summation formula is generalized to a modified form by expanding the periodic Bernoulli polynomials as its Fourier series and taking cuts, which includes both the Euler-Maclaurin summation formula and the Poission summation formula as special cases. By making use of the modified formula, a numerical summation method is obtained and the error can be controlled. The modified formula is also generalized from one dimention to two dimentions. Examples of its applications in statistical physics are also discussed.
\end{abstract}
\section{Introduction} 
Summation is widely used in physics, especially in statistical physics. A traditional but powerful tool to work out a summation is the Euler-Maclaurin formula~\cite{Abramowitz,TomM}, which has lots of generalizations~\cite{rzadkowski,trigub,Diran}. If a real function $f$ is $C^n$ on an interval $[a', b]$ for some positive integer $n$ and integers $a'<b$, then the Euler-Maclaurin formula reads
\begin{eqnarray}
  \sum_{i=a'}^{b}f(i)&=& \int_{a'}^{b} f(x)\dd x +\frac{f(a')+f(b)}{2}-\sum_{r=1}^{\nhi}\frac{B_{2r}}{(2r)!}M_{2r-1}(a',b)+R_n.
  \label{euler}
\end{eqnarray} 
Here $B_{2r}$ are the $2r$-th Bernoulli number and 
the symbols $\nhi$\ and $M$ are defined as
\begin{eqnarray}
   \nhi&=& \left[\frac{n}{2}\right],\label{eq_nhi}\\
   M_{\nu}(a',b)&=& f^{(\nu)}(a')-f^{(\nu)}(b),\nonumber
\end{eqnarray}
respectively, where the brakets in Eq.~(\ref{eq_nhi}) and throughout stand for the floor function.
The remainder term is
\begin{eqnarray} 
  R_n&=& \frac{(-1)^{n+1}}{n!}\int_{a'}^{b}P_{n}(x)f^{(n)}(x)\dd x,
  \label{Rn}
\end{eqnarray}
with $P_{n}$ being the periodic Bernoulli polynomial~\cite{Lehmer}.
Dropping the remainder term $R_n$ in Eq.~(\ref{euler}), one obtains an approximation of the summation on the left hand side. It is widely used for its high efficiency in most cases~\cite{karshon2003euler,kukhtin1991lattice,elliott1998euler}.  However, as most asymptotic expansions, for given $a', b$ and $f$, the remainder term usually does not converge to $0$ as $n\rightarrow+\infty$. In our recent work~\cite{Guo:2018qxg}, a modified Euler-Maclaurin  formula (MEMF) was introduced to keep part of the remainder term 
 in order to achieve higher accuracy. However, the MEMF was expressed 
rather technically than rigorously there. In this paper, we first review the MEMF in Section \ref{sec:md1} in a more rigorous way including an estimation of the upper limit of the remainder term, and then generalize the result to Dimension $2$ in Section \ref{md2}. Some applications of the formulae in statistical physics are given in Section \ref{app}.
\section{Modified Euler-Maclaurin Formula in Dimension 1}   \label{sec:md1}
\begin{theorem}[MEMF in 1D]\label{th_mem_1d}
   If a real function $f$ is $C^{n}$ on $[a, b]$ for some $n \in \mathbb{N}^+$ and $a$, $b \in \mathbb{Z}$ with $a<b$, then for given $m,p\in \mathbb{N}$ with $m\leq b-a$, the summation of $f$ over integer points on $[a,b]$ can be evaluated by the formula
\begin{eqnarray}
   \sum_{i=a}^{b}f(i)&=& \sum_{i=a}^{a'-1}f(i)+\int_{a'}^{b}\frac{\sin((2p+1)\pi x)}{\sin(\pi x)}f(x) \dd x+\frac{f(a')+f(b)}{2}+\sum_{r=1}^{\nhi}(-1)^{r}T_{2r,p}M_{2r-1}(a',b)+R_{mnp}\label{eq_MEL1D}\label{eq_mem_1d}
\end{eqnarray}
with $a'=a+m$ and $T_{s,p}=\frac{2}{(2\pi)^{s}}\left(\zeta (s)-\sum\limits_{k=1}^p\frac{1}{k^{s}}\right)$, where the remainder term is
\begin{eqnarray}
  R_{mnp}&=& (-1)^{n}\int_{a'}^{b}\sum_{|k|>p}\frac{\ee^{\ii 2\pi  kx}}{(\ii 2\pi k)^{n}}f^{(n)}(x){\rm d}x.
  \label{Rnd1}
 \end{eqnarray}
 The $\zeta(s)$\ in the definition of $T_{s,p}$ is the Riemann zeta function.
\end{theorem}
\begin{proof}
   Two modifications are made to the original Euler-Maclaurin formula. (i) we sum over the first $m$ terms explicitly, which is trival but necessary for some quantum systems at low temperature. (ii) we keep part of the original remainder in our calculations to improve the accuracy of the summation. 
   
   We apply the Euler-Maclaurin formula to the summation on the left hand side after modification (i) and obtain
\begin{eqnarray}  
   \sum_{i=a}^{b}f(i)&=& \sum_{i=a}^{a'-1}f(i)+\int_{a'}^{b} f(x)\dd x +\frac{f(a')+f(b)}{2}-\sum_{r=1}^{\nhi}\frac{B_{2r}}{(2r)!}M_{2r-1}(a',b)+R_{mn}\label{eq_7}
\end{eqnarray}
with $a'=a+m$.
The periodic Bernoulli polynomials of the remainder term in Eq.~(\ref{Rn}) can be expanded via Fourier series~\cite{Luo2010fourier} 
\begin{eqnarray} 
   P_{n}(x)&=& -n!\sum_{k\in Z-\{0\}}\frac{\ee^{\ii 2\pi kx}}{(\ii 2\pi k)^n}.\label{eq_pBp}
\end{eqnarray}
Taking $2p$ terms from the remainder term with $|k|\leq p$ in Eq.~(\ref{eq_pBp}) and assigning the others to a new remainder term $R_{mnp}$, we have
\begin{eqnarray*} 
   R_{mn}&=& 
  (-1)^{n}\int_{a'}^{b}
   \sum_{k\in Z-\{0\}}\frac{\ee^{\ii 2\pi kx}}{(\ii 2\pi k)^n}
  f^{(n)}(x)\dd x 
  = 
  (-1)^{n}
  \int_{a'}^{b}\sum_{0<|k|\le p}\frac{\ee^{\ii 2\pi kx}}{(\ii 2\pi k)^n}
  f^{(n)}(x)\dd x+
  R_{mnp}.
\end{eqnarray*}
Using integration by parts, we have
\begin{eqnarray}
   R_{mn}
  &=& 
  \sum_{l=0}^{n-1}\sum_{0<|k|\le p}\left(\frac{\ii}{2\pi k}\right)^{l+1}M_l(b,a')+ \int_{a'}^b\sum_{0<|k|\le p} \ee^{\ii 2\pi  k x} f(x)\dd x +  R_{mnp}\nonumber\\
  &=& \sum_{r=1}^{\nhi}\sum_{0<k\le p}\frac{2(-1)^{r+1}}{(2\pi k)^{2r}}M_{2r-1}(a',b)+\int_{a'}^b\left(\frac{\sin( (2p+1) \pi x)}{\sin(\pi x)}-1\right)  f(x)\dd x +  R_{mnp}.\label{eq_11}
\end{eqnarray}
Note that
\begin{eqnarray*}
   B_{2r}&=& 2(-1)^{r+1}\frac{\zeta(2r)(2r)!}{(2\pi)^{2r}}.
\end{eqnarray*}
Substituting Eq.~(\ref{eq_11}) into Eq.~(\ref{eq_7}), one obtains Eq.~(\ref{eq_MEL1D}), and the proof is finished.

\end{proof}
Dropping the remainder term $R_{mnp}$\ in Eq.~(\ref{eq_MEL1D}), an $m$-$n$-$p$ cut of the Modified Euler-Maclaurin formula is obtained, which can be used as an approximation of the summation on the left hand side. In practice, the integral in Eq.~(\ref{eq_MEL1D}) can also be calculated by
\begin{eqnarray}
   \int_{a'}^{b}f(x)\frac{\sin((2p+1)\pi x)}{\sin (\pi x) }{\rm d}x&=&\sum_{k=-p}^{p}\int_{a'}^{b}f(x)\ee^{\ii 2\pi kx}{\rm d}x,\label{eq_grating}
\end{eqnarray} 
whose right hand side is a summation of Fourier coefficients of $f$ on $[a,b]$. We remark that the integral kernel $\frac{\sin((2p+1)\pi x)}{\sin (\pi x)}$\ is identical to the multi-splits interference factor of an optical grating. 
\begin{corollary}
   If the conditions in Theorem~\ref{th_mem_1d} hold for all $b>a'$, and all the terms in Eq.~(\ref{eq_MEL1D}) are finite as $b\rightarrow+\infty$, then
\begin{eqnarray}
   \sum_{i=a}^{+\infty}f(i)&=& \sum_{i=a}^{a'-1}f(i)+\int_{a'}^{+\infty}\frac{\sin((2p+1)\pi x)}{\sin(\pi x)}f(x) \dd x+\frac{f(a')}{2}+\sum_{r=1}^{\nhi}(-1)^{r}T_{2r,p}M_{2r-1}(a',+\infty)+R_{mnp},\label{eq_MEL1D_inf}
\end{eqnarray}
with
\begin{eqnarray*}
  R_{mnp}&=& (-1)^{n}\int_{a'}^{+\infty}\sum_{|k|>p}\frac{\ee^{\ii 2\pi kx}}{(\ii 2\pi k)^{n}}f^{(n)}(x){\rm d}x.
  \end{eqnarray*}
\end{corollary}
Obviously, similar conclusions can be drawn for $a'=-\infty$, and for other results in this paper.
\begin{remark}
We reveal the following connections to the Euler-Maclaurin formula and to the Poisson summation formula.       
   \begin{enumerate}
	\item Taking $m=0, p=0$ in Theorem~\ref{th_mem_1d}, Eq.~(\ref{eq_MEL1D}) becomes the original Euler-Maclaurin formula, that is Eq.~(\ref{euler}).
	\item Taking $m=0, n=1$ in Theorem~\ref{th_mem_1d}, if 
	   all the terms in Eq.~(\ref{eq_MEL1D}) converges to finite values  and $R_{01p}$ converges to $0$ as $p\rightarrow +\infty$, then with the help of Eq.~(\ref{eq_grating}), Eq.~(\ref{eq_MEL1D}) becomes
\begin{eqnarray*}
   \sum_{i=a}^{b}f(i)&=& \sum_{k=-\infty}^{+\infty}\int_{a}^{b}f(x)\ee^{\ii 2\pi kx} \dd x+\frac{f(a)+f(b)}{2}.
\end{eqnarray*}
Furthermore, if all the terms in the above equation are finite and the summation on the left hand side converges as $a\rightarrow -\infty$ and $b\rightarrow +\infty$, then it becomes the Poisson summation formula
\begin{eqnarray*}
   \sum_{i=-\infty}^{+\infty}f(i)&=& \sum_{k=-\infty}^{+\infty}\int_{-\infty}^{+\infty}f(x)\ee^{\ii 2\pi kx} \dd x.
\end{eqnarray*}
   \end{enumerate}
\end{remark}
\begin{theorem}[Estimation of the remainder term]\label{th_es_1d}
   If the conditions in Theorem~\ref{th_mem_1d} hold, then the remainder term $R_{mnp}$\ in Eq.~(\ref{Rnd1}) with $n>1$ can be controlled by the following estimations.
   \begin{enumerate}
	\item
\begin{eqnarray*}
	|R_{mnp}|&\leq& R_{mnp}^A\equiv T_{n,p}\int_{a'}^b\abs{f^{(n)}(x)}\dd x.
\end{eqnarray*}
	\item
	   If $n$ is odd, and $f^{(n)}(x)$ is monotone with respect to $x$ on $[a',b]$, then 
	   \begin{eqnarray*}
		\abs{R_{mnp}}&\leq&R_{mnp}^B\equiv{2T_{n+1,p}}\abs{M_{n}(a',b)}.
	   \end{eqnarray*}
	   If $n$ is even, and $f^{(n)}(x)$ is monotone with respect to  $x$ on $[a',b]$, then 
	   \begin{eqnarray*}
		\abs{R_{mnp}}&\leq&R_{mnp}^C\equiv{T_{n+1,p}}\abs{M_{n}(a',b)}.
	   \end{eqnarray*}
   \end{enumerate}
\end{theorem}
\begin{proof}
   \ \\
   \begin{enumerate}
	\item Compute
   \begin{eqnarray*}
	\abs{R_{mnp}}&=& 
	\abs{ (-1)^{n}\int_{a'}^{b}\sum_{|k|>p}\frac{\ee^{\ii 2\pi  kx}}{(\ii 2\pi k)^{n}}f^{(n)}(x)\dd x}
	\leq
	\int_{a'}^{b}\sum_{|k|>p}\abs{\frac{\ee^{\ii 2\pi  kx}}{(\ii 2\pi k)^{n}}f^{(n)}(x)}\dd x= T_{n,p}\int_{a'}^{b}\abs{f^{(n)}(x)}\dd x.
   \end{eqnarray*}
\item For odd $n$, we estimate
   \begin{eqnarray*}
	\abs{R_{mnp}}
	&=& \abs{ (-1)^{n}\int_{a'}^{b}\sum_{|k|>p}\frac{\ee^{\ii 2\pi  kx}}{(\ii 2\pi k)^{n}}f^{(n)}(x)\dd x}\\
	&\leq& \frac{1}{(2\pi)^n}\sum_{k>p} \frac{2}{k^n}\abs{\int_{a'}^{b}\sin(2\pi k x)f^{(n)}(x)\dd x}\\
	&=& \frac{1}{(2\pi)^n}\sum_{k>p} \frac{2}{k^n}\abs{\sum_{j=1}^{{2k(b-a')}}\int_{a'+\frac{j-1}{2k}}^{a'+\frac{j}{2k}}\sin(2\pi k x)\left(f^{(n)}(x)-f^{(n)}(b)\right)\dd x}.
   \end{eqnarray*}
   Note that the summation is an alternating summation, we have
   \begin{eqnarray*}
     \abs{R_{mnp}}	&\leq& \frac{1}{(2\pi)^n}\sum_{k>p} \frac{2}{k^n}\abs{\int_{a'}^{a'+\frac{1}{2k}}\sin(2\pi k x)\left(f^{(n)}(x)-f^{(n)}(b)\right)\dd x}\\
	&\leq& \frac{1}{(2\pi)^n}\sum_{k>p} \frac{2}{k^n}\abs{f^{(n)}(a')-f^{(n)}(b)}\frac{1}{\pi k}\\
	&=& 2T_{n+1,p}\abs{M_{n}(a',b)}.
   \end{eqnarray*}
   Similarly, for even $n$, we estimate
   \begin{eqnarray*}
	\abs{R_{mnp}}
	&=& \abs{ (-1)^{n}\int_{a'}^{b}\sum_{|k|>p}\frac{\ee^{\ii 2\pi  kx}}{(\ii 2\pi k)^{n}}f^{(n)}(x)\dd x}\\
	&\leq& \frac{1}{(2\pi)^n}\sum_{k>p} \frac{2}{k^n}\abs{\int_{a'}^{b}\cos(2\pi k x)f^{(n)}(x)\dd x}\\
	&=& \frac{1}{(2\pi)^n}\sum_{k>p} \frac{2}{k^n}\abs{\left(\int_{a'}^{a'+\frac{1}{4k}}+\sum_{j=1}^{{2k(b-a')}-1}\int_{a'+\frac{j-1/2}{2k}}^{a'+\frac{j+1/2}{2k}}+\int_{b-\frac{1}{4k}}^b\right)\cos(2\pi k x)\left(f^{(n)}(x)-f^{(n)}(b)\right)\dd x}\\
	&\leq& \frac{1}{(2\pi)^n}\sum_{k>p} \frac{2}{k^n}\int_{a'}^{a'+\frac{1}{4k}}\cos(2\pi k x)\abs{f^{(n)}(x)-f^{(n)}(b)}\dd x\\
	&\leq& \frac{1}{(2\pi)^n}\sum_{k>p} \frac{2}{k^n}\abs{f^{(n)}(a')-f^{(n)}(b)}\frac{1}{2\pi k}\\
	&=& T_{n+1,p}\abs{M_{n}(a',b)}.
   \end{eqnarray*}
   \end{enumerate}
\end{proof}
\begin{corollary}\label{co_TLM}
   If the conditions in Theorem~\ref{th_mem_1d} hold, and the interval $[a',b]$ can be devided into $l_{n-1}$\ subintervals on each of which $f^{(n-1)}$\ is monotone, then the remainder term in Eq.~(\ref{eq_mem_1d}) with $n>1$ can be controlled as
   \begin{eqnarray*}
	\abs{R_{mnp}}&\leq&\tilde{R}_{mnp}^A\equiv T_{n,p}l_{n-1}\Delta M_{n-1},
   \end{eqnarray*}
   with $\Delta M_{n-1}=M_{max}-M_{min}$, where $M_{max}$ and $M_{min}$\ are the maximum and minimum values of $f^{(n-1)}$ on $[a',b]$, respectively.
\end{corollary}
\section{Generalization to Dimension 2} \label{md2} 
\begin{lemma}[MEMF in a rectangle]\label{le_mem_2d}
Let D be the rectangle 
\begin{eqnarray*}
  D=\{(x,y)\in \mathbb{R}^2| a\leqslant x <b, c\leqslant y<d\},
\end{eqnarray*}
\begin{figure}[!hbt]
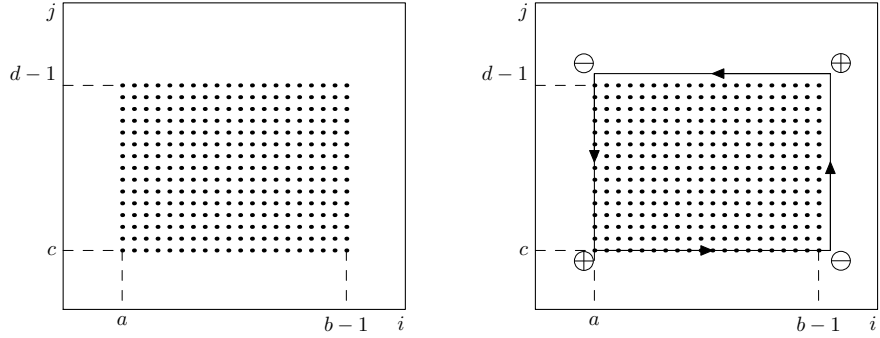
 
    \centering
    \includegraphics[width=0.3\textwidth]{rec_1}\rule{1cm}{0mm}%
    \includegraphics[width=0.3\textwidth]{rec_3}
    \caption{Summation points (left) in Lemma \ref{le_mem_2d} and the corresponding region $D$ with its boundary (right).}
    \label{fg_dog_1}
 \end{figure}
illustrated as in the left panel of Fig.~\ref{fg_dog_1}, where $a, b, c$, and $d$ are integers satisfying $a<b$, $c<d$.
If a function $f: \bar{D}\to \mathbb{R}$ with $\bar{D}=D\cup\partial D$ is $C^{n+n'}$ for some positive integers $n,n'$, then for given non-negative integers $m,m',p,p'$, the summation of $f(x,y)$ on all integer pairs $(i,j) \in D$ can be expressed in the following form
\begin{eqnarray}
   \sum_{\rm{integer\ pairs\ }(i,j) \in D} f(i,j) &=& A+L+V+R,\label{eq_sDp}
\end{eqnarray}
with
\begin{eqnarray}
	A&=& \int_D\dd x\dd y\ a_p^xa_{p'}^yf(x,y),\label{eq_def_A}\\
L&=& \int_{\partial D} -a_p^xB_{n',p'}^y f(x,y){\dd x}+a_{p'}^yB_{n,p}^xf(x,y){\dd y},\\
   V&=&\sum_{(i,j)\in\rm{vertices\ of\ }D} s_{i,j}B_{n,p}^xB_{n',p'}^yf(i,j),\label{eq_V}\\
   R&=& R_{np,n'p'}=\int_D \dd x \dd y\ \left(\left(a_p^x+b_{n,p}^x\right)c_{n',p'}^y+\left(a_{p'}^y+b_{n',p'}^y\right)c_{n,p}^x+c_{n,p}^xc_{n',p'}^y\right) f(x,y),\label{eq_mem_2d_R}
\end{eqnarray} 
and the linear operators
\begin{eqnarray} 
   B_{n,p}^x&=& -\frac{1}{2}+\sum_{r=1}^{\nhi}(-1)^{r+1}T_{2r,p}\frac{\partial^{2r-1}}{\partial x^{2r-1}},\\
   a_p^x&=& \frac{\sin\left((2p+1)\pi x\right)}{\sin(\pi x)},\\
   b_{n,p}^x&=&\frac{\partial}{\partial x}B_{n,p}^x= -\frac{1}{2}\frac{\dd}{\dd x}+\sum_{r=1}^{\nhi}(-1)^{r+1}T_{2r,p}\frac{\partial^{2r}}{\partial x^{2r}},\\
   c_{n,p}^x&=& (-1)^{n}\left(\sum_{|k|>p}\frac{\ee^{\ii 2\pi  k x}}{(\ii 2\pi k)^n}\right)\frac{\partial^{n}}{\partial x^{n}}.
   \label{eq2}
\end{eqnarray}
The operators with superscript $^y$ are defined similarly. 
The factor $s_{i,j}$\ in Eq.~(\ref{eq_V}) is $+1$\ for the vertices $(a,c)$\ and $(b,d)$, while it is $-1$\ for the vertices $(a,d)$\ and $(b,c)$ as illustrated in the right panel of Fig.~\ref{fg_dog_1}.
\end{lemma}
\begin{proof}
  We rewrite Eq.~(\ref{eq_MEL1D}) with $m=0$ as 
\begin{eqnarray}
   \sum_{i=a}^{b-1}f(i)
   &=& \int_{a}^{b}\left(a_p^x+b_{n,p}^x\right)f(x)\dd x+\int_{a}^{b}c_{n,p}^xf(x)\dd x.
\label{eq1}
\end{eqnarray} 
Applying the formula in Eq.~(\ref{eq1}) twice, it is easy to obtain
\begin{eqnarray*}
   \sum_{i=a}^{b-1}\sum_{j=c}^{d-1}f(i,j)
   &=& \int_D \left(a_p^x+b_{n,p}^x\right)\left(a_{p'}^y+b_{n',p'}^y\right)f(x,y)\dd x \dd y + R_{np,n'p'}\\
   &=& \int_D a_p^xa_{p'}^yf(x,y)\dd x \dd y + \int_D \left(a_p^xb_{n',p'}^y+b_{n,p}^xa_{p'}^y\right)f(x,y)\dd x\wedge \dd y\\
   & &+ \int_D b_{n,p}^xb_{n',p'}^yf(x,y)\dd x \dd y + R_{np,n'p'}\\
   &=& A + \int_D \dd \left(-a_p^xB_{n',p'}^yf(x,y)\dd x+B_{n,p}^xa_{p'}^yf(x,y)\dd y\right)\\
   & &+ \int_D \frac{\partial}{\partial x}\frac{\partial}{\partial y}B_{n,p}^xB_{n',p'}^yf(x,y)\dd x \dd y + R\\
   &=& A + L  + V + R.
\end{eqnarray*}
That is Eq.~(\ref{eq_sDp}).
\end{proof}
Note that (i) Eq.~(\ref{eq_sDp}) holds for a single point, that is the case with $b=a+1$ and $d=c+1$, and (ii) all terms in Eq.~(\ref{eq_sDp}) are additive with respect to the region $D$. If the region $D$\ is composed of several rectangles instead, Eq.~(\ref{eq_sDp}) still holds, as long as $s_{i,j}$ is properly redefined. This is the content of the next theorem.
\begin{theorem}[MEMF in 2D]\label{co_mem_2d}
Suppose $P$\ is a set of integer pairs with $P_0\subset P$, $P_1=P-P_0$, and 
\begin{eqnarray}
	D&=& \bigcup_{(i,j)\in P_1}\{(x,y)\in \mathbb{R}^2| i\leqslant x<i+1, j\leqslant y<j+1\}.\label{eq_D}
\end{eqnarray}
Under the conditions of Lemma~\ref{le_mem_2d}, we have
\begin{eqnarray}
   \sum_{(i,j)\in P}f(i,j)&=& \sum_{(i,j)\in P_0}f(i,j)+A+L+V+R,\label{eq_mem_2d_co}
\end{eqnarray}
if all the terms are finite and well-defined, where $A$, $L$, $V$, and $R$\ are defined by Eqs.~(\ref{eq_def_A})-(\ref{eq2}), except that $s_{i,j}$\ is $+1$ if we turn from a vertical direction to a horizontal direction at $(i,j)$ along the boundary $\partial D$, and it is $-1$ if we turn from a horizontal direction to a vertical direction at $(i,j)$ along $\partial D$, such that the plus and minus signs appear alternately.
\end{theorem}
  Again, dropping the remainder term $R$, an approximation of the original summation in Dimension 2 is obtained. 
  \begin{theorem}[Estimation of the remainder term in 2D]
	  If $f$ satisfies the conditions in Theorem~\ref{co_mem_2d}, then the remainder term in Eq.~(\ref{eq_mem_2d_co}) [i.e. Eq.~(\ref{eq_mem_2d_R})] with $n, n'>1$ can be controlled by the following estimation
\begin{eqnarray}
   |R_{np,n'p'}|\leq R_{np,n'p'}^A&\equiv &
	T_{n,p}\left( (2p'+1)\int_D\ \abs{f_{n,0}(x,y)}\dd x \dd y+\frac{1}{2}\int_{\partial D}\abs{f_{n,0}(x,y)\dd x}\right)\nonumber\\
	&&+T_{n',p'}\left( (2p+1)\int_D \ \abs{f_{0,n'}(x,y)}\dd x \dd y+\frac{1}{2}\int_{\partial D}\abs{f_{0,n'}(x,y)\dd y}\right)\nonumber\\
	&&+T_{n,p}\sum_{r=1}^{\hnp}T_{2r,p'}\int_{\partial D} \abs{f_{n,2r-1}(x,y)\dd x}+T_{n',p'}\sum_{r=1}^{\hn}T_{2r,p}\int_{\partial D}\abs{f_{2r-1,n'}(x,y)\dd y}\nonumber\\
	&&+T_{n,p}T_{n',p'}\int_D \ \abs{f_{n,n'}(x,y)}\dd x \dd y
	\label{eq_mem_th4}
\end{eqnarray}
with $f_{\mu,\nu}(x,y)\equiv\frac{\partial^{\mu+\nu}}{\partial x^{\mu}\partial y^{\nu}}f(x,y)$.
\end{theorem}
\begin{proof}
   Eq.~(\ref{eq_mem_2d_R}) can be estimated by
   \begin{eqnarray}
	\abs{R_{np,n'p'}}
	&=& \abs{\int_D \ \left(\left(a_{p'}^y+b_{n',p'}^y\right)c_{n,p}^x+\left(a_p^x+b_{n,p}^x\right)c_{n',p'}^y+c_{n,p}^xc_{n',p'}^y\right) f(x,y)\dd x \dd y}\nonumber\\
	&\leq& \int_D  \left(\abs{a_{p'}^yc_{n,p}^xf(x,y)}+\abs{a_p^xc_{n',p'}^yf(x,y)}+\abs{c_{n,p}^xc_{n',p'}^yf(x,y)}\right)\dd x \dd y\ \nonumber\\
	&&+ \abs{\int_D \ \left(b_{n',p'}^yc_{n,p}^xf(x,y)+b_{n,p}^xc_{n',p'}^yf(x,y)\right)\dd x \dd y}.
   \end{eqnarray}
   Note that, following the procedure in the proof of Theorem~\ref{th_es_1d}, we have
   \begin{eqnarray}
	\abs{c_{n,p}^xf(x,y)}&\le&T_{n,p}\abs{f_{n,0}(x,y)},\label{eq_es_cx}\\
	\abs{c_{n',p'}^yf(x,y)}&\le&T_{n',p'}\abs{f_{0,n'}(x,y)}\label{eq_es_cy}.
   \end{eqnarray}
   Thus,
   \begin{eqnarray*}
	&& \int_D\ \left(\abs{a_{p'}^yc_{n,p}^xf(x,y)}+\abs{a_p^xc_{n',p'}^yf(x,y)}+\abs{c_{n,p}^xc_{n',p'}^yf(x,y)}\right) \dd x \dd y \nonumber\\
	&\leq& T_{n,p}\int_D\ \abs{a_{p'}^yf_{n,0}(x,y)}\dd x \dd y+T_{n',p'}\int_D \ \abs{a_p^xf_{0,n'}(x,y)}\dd x \dd y+T_{n,p}T_{n',p'}\int_D \ \abs{f_{n,n'}(x,y)}\dd x \dd y.
   \end{eqnarray*}
   \begin{eqnarray*}
	&& \abs{\int_D \ \left(b_{n',p'}^yc_{n,p}^xf(x,y)+b_{n,p}^xc_{n',p'}^yf(x,y)\right)\dd x \dd y}\\
	 &=&\abs{\int_D \dd \left(-B_{n',p'}^yc_{n,p}^xf(x,y)\dd x+B_{n,p}^xc_{n',p'}^yf(x,y)\dd y\right)} \\
	&=&\abs{\int_{\partial D} \left(-B_{n',p'}^yc_{n,p}^xf(x,y)\dd x+B_{n,p}^xc_{n',p'}^yf(x,y)\dd y\right)} \\
	&\leq&T_{n,p}\int_{\partial D}\abs{B_{n',p'}^yf_{n,0}(x,y)\dd x}+T_{n',p'}\int_{\partial D}\abs{B_{n,p}^xf_{0,n'}(x,y)\dd y}.
   \end{eqnarray*}
   Note that
   \begin{eqnarray*}
	\abs{a_p^x}&\leq&2p+1,\\
	\abs{B_{n,p}^x f(x,y)}&=& \abs{\left(-\frac{1}{2}+\sum_{r=1}^{\nhi}(-1)^{r+1}T_{2r,p}\frac{\partial^{2r-1}}{\partial x^{2r-1}}\right)f(x,y)}
	\leq\frac{1}{2}\abs{f(x,y)}+\sum_{r=1}^{\hn}T_{2r,p}\abs{f_{2r-1,0}(x,y)}.
   \end{eqnarray*}
   Similarly,
   \begin{eqnarray*}
	\abs{a_{p'}^y}&\leq&2p'+1,\\
	\abs{B_{n',p'}^y f(x,y)}&=& \abs{\left(-\frac{1}{2}+\sum_{r=1}^{\hnp}(-1)^{r+1}T_{2r,p'}\frac{\partial^{2r-1}}{\partial y^{2r-1}}\right)f(x,y)} 
	\leq\frac{1}{2}\abs{f(x,y)}+\sum_{r=1}^{\hnp}T_{2r,p'}\abs{f_{0,2r-1}(x,y)}.
   \end{eqnarray*}
   Thus Eq.~(\ref{eq_mem_th4}) holds.
\end{proof}
Obviously, for given $n$, $n'$ and $f$, one can always choose proper $p$ and $p'$ to achieve as high accuracy as expected.
For $n'=n$ and $p'=p$, the formula is simpler as follows.
\begin{corollary}[Estimation of the remainder term with $n=n'$ and $p=p'$]\label{co_rem2dnp}
   If $f$ satisfies the conditions in Theorem~\ref{co_mem_2d}, then
\begin{eqnarray*}
   |R_{np,np}|&\leq&R_{np,np}^{A}\\
   &=& T_{n,p}\left( (2p+1)\int_D\ \left(\abs{f_{n,0}(x,y)}+\abs{f_{0,n}(x,y)}\right)\dd x \dd y+\frac{1}{2}\int_{\partial D}\abs{f_{0,n}(x,y)\dd y}+\abs{f_{n,0}(x,y)\dd x}\right)\nonumber\\
	&&+T_{n,p}\sum_{r=1}^{\hn}T_{2r,p}\left(\int_{\partial D} \abs{f_{n,2r-1}(x,y)\dd x}+\abs{f_{2r-1,n}(x,y)\dd y}\right)+T_{n,p}^2\int_D \ \abs{f_{n,n}(x,y)}\dd x \dd y.
\end{eqnarray*}
\end{corollary}
\section{Applications}  \label{app} 
\subsection{Partition function of a one-dimensional infinite square well}
For a one-dimensional infinite square well, the eigen energy is 
\begin{eqnarray*} 
	\epsilon_l&=& \epsilon l^2,\qquad l\in \mathbb{N}^+,
\end{eqnarray*}
where $\epsilon=\epsilon_1$\ is the eigen energy of the ground state. Thus the partition function
\begin{eqnarray*}
   Z&=& \sum_{l=1}^{+\infty}f(l),
\end{eqnarray*}
   with $f(x)=\ee^{-Bx^2}$ and $B=\beta\epsilon_1$. 
   Then Eq.~(\ref{eq_MEL1D_inf}) gives
\begin{eqnarray*}
   Z &=& Z_m+W_{np}(m+1)+R_{mnp},
 \end{eqnarray*}
 with
 \begin{eqnarray}
	 Z_m&=& \sum_{l=1}^me^{-Bl^2},\label{eq_62}\\
	 W_{np}(x)&=& T_p(x)+U_{np}(x),\label{eq_63}\\
	 T_p(x)&=&\int_x^{+\infty}\frac{\sin((2p+1)\pi \xi)}{\sin(\pi \xi)}f(\xi)\dd \xi= \frac{1}{2}\sqrt{\frac{\pi}{B}}\sum_{k=-p}^pe^{-\frac{\pi^2k^2}{B}}\erfc\left(\sqrt{B}x-\frac{i\pi k}{\sqrt{B}}\right),\nonumber\\
	 U_{np}(x)&=& \frac{1}{2}f(x)+\sum_{r=1}^{\nhi}(-1)^{r}T_{2r,p}f^{(2r-1)}(x)  =\frac{1}{2}\ee^{-Bx^2}+\sum_{r=1}^{\nhi}(-1)^{r+1}T_{2r,p}B^{r-\frac{1}{2}}H_{2r-1}(\sqrt{B}x)\ee^{-Bx^2},\nonumber
 \end{eqnarray}
 where $\erfc(x)=\frac{2}{\sqrt{\pi}}\int_x^{+\infty}\ee^{-z^2}\dd z$\ is the complementary error function, and $H_n(x)=(-1)^ne^{x^2}\frac{\dd^n}{\dd x^n}\ee^{-x^2}$\ is the Hermite polynomial.
 It can be proved that
 \begin{eqnarray*}
   \abs{H_n(x)\ee^{-x^2}}&\leq&\frac{2^n\Gamma(\frac{n+1}{2})}{\sqrt{\pi}}
\end{eqnarray*}
for all $x \in \mathbb{R}$ and all $n\in\mathbb{N}$, and $H_n(x)\ee^{-x^2}$\ has at most $\left[\frac{n+3}{2}\right]$ monotonic intervals on $[m+1,+\infty)$,
the remainder term can be controlled according to Corollary~\ref{co_TLM} by
\begin{eqnarray}
	\abs{R_{mnp}}\leq \tilde{R}_{mnp}^A\leq \bar{R}_{np}^A\equiv T_{n,p}\frac{\left(\nhi+1\right)2^{n}\Gamma(\frac{n}{2})}{\sqrt{\pi}}B^{\frac{n-1}{2}}\label{eq_67}
\end{eqnarray}
for $n>1$.

To have a better estimation of the remainder term, especially when $m$ is large, we need a better estimation of $H_n(x)$. Note that $H_n(x)\ee^{-x^2/2}$\ is the (unnormalized) wave function of a quantum oscillator. By considering the classical corresponding of a quantum oscillator~\cite{Pauling}, we have the following conjecture
\begin{conjecture}\label{conj}
	An upper bound of Hermite polynomials can be estimated by
\begin{eqnarray}
   \abs{H_n(x)\ee^{-x^2}}&\leq&\frac{2^n\Gamma(\frac{n+1}{2})}{\sqrt{\pi}}g_n(\textrm{min}(x,x_n))\ee^{-x^2/2},\label{hyp}
\end{eqnarray}
with
\begin{eqnarray*}
	g_n(x)&=& \left({1-\frac{x^2}{2n+1}}\right)^{-1/2},
\end{eqnarray*}
and $x_n=\sqrt{2n+1}\left(1-\frac{\pi}{2(2n+1)}\right)$ for all $n\in \mathbb{N}^+$, and $x>0$. 
\end{conjecture}
We have checked numerically that the conjecture holds at least for $n\leq 20$, and one can easily check given reasonable larger $n$ numerically when necessary.  As an illustration, the first $9$ values of $x_n$ and $g_n(x_n)$\ are listed in Table~\ref{tb_xngn}. If Eq.~(\ref{hyp}) holds, there is a rougher but simpler estimation:
\begin{eqnarray}
   \abs{H_n(x)\ee^{-x^2}}&\leq&G_ne^{-x^2/2},\label{hyp_2}
\end{eqnarray}
with
\begin{eqnarray}
	G_n&\equiv&\frac{2^n\Gamma(\frac{n+1}{2})}{\sqrt{\pi}}g_n(x_n)=\frac{2^{n+1}(2n+1)\Gamma(\frac{n+1}{2})}{\pi\sqrt{4(2n+1)-\pi}},
\end{eqnarray}
for all $n\in\mathbb{N}^+$, since $g_n(x)$ increases as $x$ increases at $x>0$.
Then the remainder can be controlled according to Corollary~\ref{co_TLM} by
\begin{eqnarray}
	\abs{R_{mnp}}&\leq&\tilde{R}_{mnp}^H\equiv T_{n,p}2(\nhi+1)G_{n-1}\ee^{-\frac{B(m+1)^2}{2}}B^{\frac{n-1}{2}} \label{eq_RH}
\end{eqnarray}
for $n>1$. The superscript $^H$ here and in the following implies that we have applied Eq.~(\ref{hyp_2}) which is based on Conjecture \ref{conj} on Hermite polynomials. From Eq.~(\ref{eq_RH}), it can be seen that when $n$\ is properly large, the remainder will be controlled for small $B$ due to the power term, and when $m$\ is properly large, the remainder will be controlled for large $B$ due to the exponential term, and when $p$\ is properly large, the remainder will be suppressed overall due to the effect of the $T_{n,p}$ term. Therefore, one can easily make the remainder term as small as possible. We list some values of $\bar{R}_{np}^A$ [defined in Eq.~(\ref{eq_67})], $\tilde{R}_{0np}^H$ [defined in Eq.~(\ref{eq_RH})] in Table~\ref{tb_bR}. It can be seen that the result can be as accurate as possible by choosing proper $m$, $n$, and $p$. If we regard $\tilde{R}^H_{mnp}$\ as  a function of $B$, the maximum is achieved at $B_m=(n-1)/(m+1)^2$, so that the remainder term can be controlled for all values of $B$.  We plot $\tilde{R}^H_{252}$  as a function of $B$ as an example in Fig.~\ref{fg_R}, where $B_m=4/9$. 
Eq.~(\ref{hyp}) can further be improved for large $x$, and we leave it for future work. 
\begin{table}[!hbt]
	\centering
   \begin{tabular}{cccccccccc}
	\hline
	$n $&$ 2 $&$ 3 $&$ 4 $&$ 5 $&$ 6 $&$ 7 $&$ 8$&$9$&$10$\\
	\hline
	$x_n$&$1.534$&$2.052$&$2.476$&$2.843$&$3.170$&$3.467$&$3.742$&$3.999$&$4.240$\\
	\hline
	$g_n(x_n)$&$1.374$&$1.584$&$1.772$&$1.942$&$2.099$&$2.245$&$2.382$&$2.512$&$2.635$\\
	\hline
   \end{tabular}
   \caption{Approximate values of the first several $x_n$ and $g_n(x_n)$.}\label{tb_xngn}
\end{table}
\begin{table}[!hbt]
	\centering
   \begin{tabular}{ccccccccccccccc}
 	\hline
	\diagbox{$p$}{$n$} &$ 3 $&$ 5 $&$ 7 $&$ 9 $&$ 11 $&$ 13 $&$ 15 $\\
	\hline
	$0$ & $7.8\times 10^{-2 }$&$ 1.5\times 10^{-2 }$&$ 5.0\times 10^{-3 }$&$ 2.2\times 10^{-3 }$&$ 1.2\times 10^{-3 }$&$ 7.8\times 10^{-4 }$&$ 5.9\times 10^{-4 }$\\
	$1$ & $1.3\times 10^{-2 }$&$ 5.4\times 10^{-4 }$&$ 4.1\times 10^{-5 }$&$ 4.4\times 10^{-6 }$&$ 6.0\times 10^{-7 }$&$ 9.6\times 10^{-8 }$&$ 1.8\times 10^{-8 }$\\
	$2$ & $5.0\times 10^{-3 }$&$ 8.3\times 10^{-5 }$&$ 2.7\times 10^{-6 }$&$ 1.2\times 10^{-7 }$&$ 7.1\times 10^{-9 }$&$ 5.0\times 10^{-10 }$&$ 4.2\times 10^{-11 }$\\
	$3$ & $2.6\times 10^{-3 }$&$ 2.3\times 10^{-5 }$&$ 3.9\times 10^{-7 }$&$ 9.8\times 10^{-9 }$&$ 3.2\times 10^{-10 }$&$ 1.2\times 10^{-11 }$&$ 5.7\times 10^{-13 }$\\
	$4$ & $1.6\times 10^{-3 }$&$ 8.6\times 10^{-6 }$&$ 9.2\times 10^{-8 }$&$ 1.4\times 10^{-9 }$&$ 2.9\times 10^{-11 }$&$ 7.1\times 10^{-13 }$&$ 2.1\times 10^{-14 }$\\
 	\hline
\end{tabular}
\caption{Approximate values of $\bar{R}_{np}^A$ [defined in Eq.~(\ref{eq_67})] with $B=1$.}\label{tb_bR}
\end{table}
\begin{table}[!hbt]
	\centering
   \begin{tabular}{ccccccccccccc}
 	\hline
	\diagbox{$p$}{$n$}&$ 3 $&$ 5 $&$ 7 $&$ 9 $&$ 11 $&$ 13 $&$ 15 $\\
	\hline
 $0 $&$ 6.5\times 10^{-2} $&$ 1.6\times 10^{-2} $&$ 6.4\times 10^{-3} $&$ 3.2\times 10^{-3} $&$ 1.9\times 10^{-3} $&$ 1.4\times 10^{-3} $&$ 1.1\times 10^{-3} $\\
 $1 $&$ 1.1\times 10^{-2} $&$ 5.8\times 10^{-4} $&$ 5.3\times 10^{-5} $&$ 6.4\times 10^{-6} $&$ 9.5\times 10^{-7} $&$ 1.7\times 10^{-7} $&$ 3.4\times 10^{-8} $\\
 $2 $&$ 4.1\times 10^{-3} $&$ 9.\times 10^{-5} $&$ 3.4\times 10^{-6} $&$ 1.8\times 10^{-7} $&$ 1.1\times 10^{-8} $&$ 8.8\times 10^{-10} $&$ 7.8\times 10^{-11} $\\
 $3 $&$ 2.2\times 10^{-3} $&$ 2.5\times 10^{-5} $&$ 5.0\times 10^{-7} $&$ 1.4\times 10^{-8} $&$ 5.1\times 10^{-10} $&$ 2.2\times 10^{-11} $&$ 1.1\times 10^{-12} $\\
 $4 $&$ 1.3\times 10^{-3} $&$ 9.3\times 10^{-6} $&$ 1.2\times 10^{-7} $&$ 2.1\times 10^{-9} $&$ 4.6\times 10^{-11} $&$ 1.2\times 10^{-12} $&$ 3.9\times 10^{-14} $\\
 	\hline
\end{tabular}
\caption{Approximate values of $\tilde{R}_{0np}^H$ [defined in Eq.~(\ref{eq_RH})] with $B=1$.}
\end{table}
\begin{figure}[!hbt]
    \centering
    \includegraphics[width=0.5\textwidth]{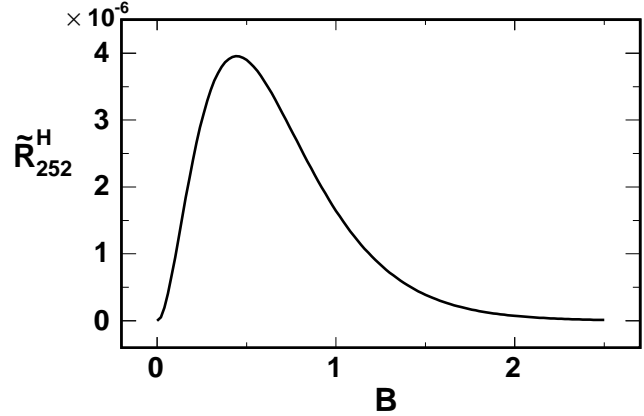}
    \caption{The error bound $\tilde{R}^H_{252}$ [defined in Eq.~(\ref{eq_RH})]\ as a function of $B$.}
    \label{fg_R}
 \end{figure}
 
\subsection{Partition function of a quantum rotator}
For a quantum rotator, the eigen energy is 
\begin{eqnarray*}
	\epsilon_l&=& \frac{l(l+1)}{2}\epsilon_1, \qquad l\in \mathbb{N}
\end{eqnarray*}
with its degeneracy $\omega_l= 2l+1$, and $\epsilon_1$ is the energy level of the first excited states.  The partition function is given by
\begin{eqnarray*}
   Z&=& \sum_{l=0}^{+\infty}\omega_l \ee^{-\beta\epsilon_l}=\sum_{l=0}^{+\infty}(2l+1)\ee^{-\mathcal{B}l(l+1)}=-\frac{\ee^{\mathcal{B}/4}}{\mathcal{B}}\sum_{l=0}^{+\infty}f'(l+{\scriptstyle\half{1}}),
\end{eqnarray*}
with $\mathcal{B}=\beta \epsilon_1/2$ and $f(x)=\ee^{-\mathcal{B} x^2}$. Compute
\begin{eqnarray*}
   Z&=& Z_m'+\mathcal{W}_{np}(m+{\scriptstyle\half{1}})+R_{mnp},\\
   Z_m'&=& \sum_{l=0}^{m-1}(2l+1)\ee^{-\mathcal{B}l(l+1)},\\
   \mathcal{W}_{np}(x)&=& \mathcal{T}_p(x)+\mathcal{U}_{np}(x),\\
   \mathcal{T}_p(x)&=& \frac{2p+1}{\mathcal{B}}\ee^{-\mathcal{B}(x^2-\frac{1}{4})}+2\left(\frac{\pi}{\mathcal{B}}\right)^{3/2}\ee^{\frac{\mathcal{B}}{4}}\sum_{k=1}^p(-1)^{k} k \ee^{-\frac{\pi^2k^2}{\mathcal{B}}} \textrm{Im}\left( \erf\left(\sqrt{\mathcal{B}}x-i\frac{k\pi}{\sqrt{\mathcal{B}}}\right)\right),\\
   \mathcal{U}_{np}(x)&=& \left(x+\sum_{r=1}^{\nhi}(-1)^{r+1}T_{2r,p}\mathcal{B}^{r-1}H_{2r}(\sqrt{\mathcal{B}}x)\right)\ee^{-\mathcal{B}(x^2-\frac{1}{4})},
\end{eqnarray*} 
where $\erf(x)=\frac{2}{\sqrt{\pi}}\int_0^x \ee^{-z^2}\dd z$\ is the error function. 
If Eq.~(\ref{hyp}) holds, then the remainder can be controlled through Eq.~(\ref{hyp_2}) according to Corollary~\ref{co_TLM} as
\begin{eqnarray}
	\abs{R_{mnp}}&\leq&\tilde{R}_{mnp}^H=T_{n,p}\left[\frac{n+3}{2}\right]2G_n \ee^{-\mathcal{B}\left(\frac{m(m+1)}{2}-\frac{1}{8}\right)}\mathcal{B}^{\frac{n}{2}-1} \label{eq_RH_ro}
\end{eqnarray}
for $n>1$. Again, one can achieve given accuracy by choosing proper $m$, $n$, and $p$. 
 \subsection{Partition function of a two-dimensional infinite square well}
 \begin{figure}[!hbt]
     \centering
     \includegraphics[width=0.3\textwidth]{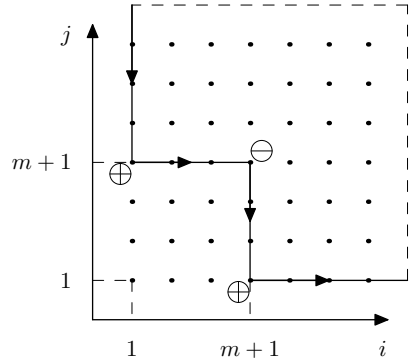}
     \caption{Summation region of the partition function of a two-dimentional square well.}
     \label{fg_sw}
\end{figure}
 We assume that a particle is in a two-dimensional infinite square well
 .  Then the eigen energies are given by
 \begin{eqnarray*}
	 \epsilon_{i,j}&=&\epsilon_{1,0}(i^2+j^2), \qquad i,j=1,2,3\ldots
 \end{eqnarray*} 
The partition function is
\begin{eqnarray*}
  Z&=&\sum_{i=1}^{+\infty}\sum_{j=1}^{+\infty}\ee^{-\beta \epsilon_{i,j}}=\sum_{i=1}^{+\infty}\sum_{j=1}^{+\infty}\ee^{-B(i^2+j^2)},
\end{eqnarray*}
with
$B=\beta\epsilon_{1,0}$. Applying Theorem~\ref{co_mem_2d} to the above summation with $P={\mathbb{N}^+}^2$, $P_0=\{(i,j)\in P|i \le m, j\le m\}$, $P_1=P-P_0$, $n'=n$, and $p'=p$, we obtain
\begin{eqnarray}
  Z=\sum_{(i,j)\in P_0}\ee^{-B(i^2+j^2)}+\sum_{(i,j) \in\ P_1}\ee^{-B(i^2+j^2)}= Z_m^2+W_{np}(m+1)\left(2W_{np}(1)-W_{np}(m+1)\right)+R_{np,np},
  \label{Z1}
\end{eqnarray}
where $Z_m$\ and $W_{np}(x)$ are defined in Eqs.~(\ref{eq_62}) and (\ref{eq_63}).

Now we try to give an upper bound of the remainder with the help of Conjecture \ref{conj}. With the help of Eq.~(\ref{hyp_2}),
we have
\begin{eqnarray}
	\abs{f_{n_1,n_2}(x,y)}&\leq& B^{\frac{n_1+n_2}{2}}G_{n_1}G_{n_2}\ee^{-B\frac{x^2+y^2}{2}}, \label{eq_82}\\
	\abs{f_{n_1,0}(x,y)}&\leq& B^{\frac{n_1}{2}}G_{n_1}\ee^{-B(\frac{x^2}{2}+y^2)}, \label{eq_83}
\end{eqnarray}
for $f(x,y)=\ee^{-B(x^2+y^2)}$ and $n_1,n_2 \in \mathbb{N}^+$.
Note that $f(x,y)=f(y,x)$. According to Corollary~\ref{co_rem2dnp}, the remainder can be controlled by
\begin{eqnarray*}
	\abs{R_{np,np} }
   \leq R_{np,np}^{A}
   &=& T_{n,p}\left( 2(2p+1)\int_D\ \abs{f_{n,0}(x,y)}\dd x \dd y+\int_{\partial D}\abs{f_{n,0}(x,y)\dd x}\right)\nonumber\\
	&&+2T_{n,p}\left(\sum_{r=1}^{\hn}T_{2r,p}\int_{\partial D} \abs{f_{n,2r-1}(x,y)\dd x}\right)+T_{n,p}^2\int_D \ \abs{f_{n,n}(x,y)}\dd x \dd y.
\end{eqnarray*}
The four integrals above can be estimated by Eqs.~(\ref{eq_82}), (\ref{eq_83}) as
\begin{eqnarray*}
   \int_D\ \abs{f_{n,0}(x,y)}\dd x \dd y&\leq&B^{\frac{n}{2}}G_n\xi(\sqrt{B/2},\sqrt{B},m+1),\\
   \int_D \ \abs{f_{n,n}(x,y)}\dd x \dd y&\leq&B^{n}G_n^2\xi(\sqrt{B/2},\sqrt{B/2},m+1),\\
   \int_{\partial D}\abs{f_{n,0}(x,y)\dd x}&\leq&B^{\frac{n}{2}}G_n\eta(\sqrt{B/2},\sqrt{B},m+1),\\
   \int_{\partial D}\abs{f_{n,2r-1}(x,y)\dd x}&\leq&B^{\frac{n-1}{2}+r}G_nG_{2r-1}\eta(\sqrt{B/2}, \sqrt{B/2},m+1),
\end{eqnarray*}
with
\begin{eqnarray*}
	\xi(x,y,\nu)&\equiv& \frac{\pi}{4xy}\left[\erfc(x)\erfc(\nu y)+\erfc(y)\erfc(\nu x)-\erfc(\nu x)\erfc(\nu y)\right],\\
	\eta(x,y,\nu)&\equiv& \frac{\sqrt{\pi}}{2x}\left[\erfc(x)\ee^{-(\nu y)^2}+\erfc\left( \nu x\right)\ee^{-y^2}-\erfc(\nu x)\ee^{-(\nu y)^2}\right].
\end{eqnarray*}
Therefore we have
\begin{eqnarray*}
	\abs{R_{np,np} }
	&\leq& \tilde{R}_{np,np}^{H}\\
	&\equiv& T_{n,p}B^{\frac{n}{2}}G_n\left( 2(2p+1)\xi(\sqrt{B/2},\sqrt{B},m+1)+\eta(\sqrt{B/2},\sqrt{B},m+1)\phantom{\sum_{r=1}^{\hn}}\right.\nonumber\\
   &&\left.+2\sum_{r=1}^{\hn}T_{2r,p}B^{r-\frac{1}{2}}G_{2r-1}\eta(\sqrt{B/2}, \sqrt{B/2},m+1)+T_{n,p}B^{\frac{n}{2}}G_n\xi(\sqrt{B/2},\sqrt{B/2},m+1)\right),
\end{eqnarray*}
where $n>1$, $p\geq 0$, and $m\geq0$. One can achieve any given accuracy by chooing $m$, $n$, $p$ properly. Some values of $\tilde{R}_{np,np}^H$ are listed in Table~\ref{tb_sw2_1}.

\begin{table}[!hbt]
	\centering
   \begin{tabular}{ccccccccccccccc}
 	\hline
	\diagbox{$p$}{$n$} &$ 3 $&$ 5 $&$ 7 $&$ 9 $&$ 11 $&$ 13 $&$ 15 $\\
	\hline
	$0 $&$ 2.2\times 10^{-2} $&$ 4.7\times 10^{-3} $&$ 1.6\times 10^{-3} $&$ 7.2\times 10^{-4} $&$ 4.0\times 10^{-4} $&$ 2.6\times 10^{-4} $&$ 2.0\times 10^{-4}$\\
	$1 $&$ 5.8\times 10^{-3} $&$ 2.6\times 10^{-4} $&$ 2.1\times 10^{-5} $&$ 2.3\times 10^{-6} $&$ 3.1\times 10^{-7} $&$ 5.1\times 10^{-8} $&$ 9.6\times 10^{-9}$\\
	$2 $&$ 3.2\times 10^{-3} $&$ 5.8\times 10^{-5} $&$ 1.9\times 10^{-6} $&$ 9.0\times 10^{-8} $&$ 5.3\times 10^{-9} $&$ 3.8\times 10^{-10} $&$ 3.2\times 10^{-11}$\\
	$3 $&$ 2.1\times 10^{-3} $&$ 2.1\times 10^{-5} $&$ 3.7\times 10^{-7} $&$ 9.5\times 10^{-9} $&$ 3.1\times 10^{-10} $&$ 1.2\times 10^{-11} $&$ 5.7\times 10^{-13}$\\
	$4 $&$ 1.6\times 10^{-3} $&$ 9.7\times 10^{-6} $&$ 1.1\times 10^{-7} $&$ 1.7\times 10^{-9} $&$ 3.5\times 10^{-11} $&$ 8.7\times 10^{-13} $&$ 2.6\times 10^{-14}$\\
 	\hline
   \end{tabular}
   \caption{$\tilde{R}_{np,np}^H$ for square well with $B=1$ and $m=0$.}\label{tb_sw2_1}
\end{table}

\section{Conclusion} 
In this paper, a modified Euler-Maclaurin summation formula (MEMF) is introduced by working out part of the Fourier expansion of the remainder term in the original Euler-Maclaurin formula, which is also generalized to 2D. The MEMF includes both the original Euler-Maclaurin formula and the Poisson summation formula as special cases, and the remainder term of the MEMF can well be under control when the parameters $m, n, p$ are properly chosen, so that it can be used as a practical numerical summation formula. Approximate expression of the partition functions of square well in 1D and 2D, and that of a quantum rotator are obtained with the help of the MEMF.
\section{Acknowledgments}
The work was supported by the ``Qinggu project'' in Tianjin University and by the NSFC under the Grant No.s 11547043, 11705125. We are grateful to Dr. Wu-Sheng Dai, Dr. Mi Xie,  Dr. Yong Zhang, Dr. Kailiang Lin, and Dr. Minghua Lin for helpful discussions.
\bibliographystyle{spphys.bst}
\bibliography{ref}

\end{document}